\documentclass[aps,pra,twocolumn,nofootinbib,showpacs,superscriptaddress]{revtex4-1}
\usepackage{amsmath}
\usepackage{amsfonts}
\usepackage{amssymb}
\usepackage{amstext}
\usepackage{textcomp}
\usepackage{amsthm}
\usepackage[makeroom]{cancel}
\usepackage{hyperref}
\usepackage{natbib}
\usepackage{graphicx}
\usepackage{bbm}
\usepackage{mathdots}
\usepackage{color}
\usepackage{multirow}
\usepackage{bm}

\input{Qcircuit}

 \def\be{\begin{equation}}
 \def\ee{\end{equation}}
 \def\bes{\begin{eqnarray}}
 \def\ees{\end{eqnarray}}
 
 \theoremstyle{plain}
\newtheorem{defi}{Definition}
\newtheorem{theorem}{Theorem}
\newtheorem{cor}{Corollary}
\newtheorem{oracle}{Data Input}

\newtheorem{result}{Result}

\newcommand{\Ord}[1]{\mathcal{O}\left( #1 \right)} 
\newcommand{\tOrd}[1]{\tilde{\mathcal{O}}\left( #1 \right)} 

\begin{document}

\title{Batched quantum state exponentiation and quantum Hebbian learning}
\author{Thomas R. Bromley}
\email{tom@xanadu.ai}
\author{Patrick Rebentrost}
\email{pr@patrickre.com}

\affiliation{Xanadu, 777 Bay St, Toronto, M5G 2C8, Canada}

\date{\today}

\begin{abstract}
Machine learning is a crucial aspect of artificial intelligence. This paper details an approach for quantum Hebbian learning through a batched version of quantum state exponentiation. Here, batches of quantum data are interacted with learning and processing quantum bits (qubits) by a series of elementary controlled partial swap operations, resulting in a Hamiltonian simulation of the statistical ensemble of the data. We decompose this elementary operation into one and two qubit quantum gates from the Clifford+$T$ set and use the decomposition to perform an efficiency analysis. Our construction of quantum Hebbian learning is motivated by extension from the established classical approach, and it can be used to find details about the data such as eigenvalues through phase estimation. 
This work contributes to the near-term development and implementation of quantum machine learning techniques.
\end{abstract}
\maketitle

\section{Introduction}

Machine learning encompasses a series of techniques that allow computers to solve problems without explicity telling them how to do so~\cite{mackay2003information}. In supervised learning, the machine is first taught to solve the problem on a series of training data. This learning stage is a crucial element in determining the performance of a machine learning algorithm. One particularly fruitful machine learning technique is to construct an artificial neural network, represented by an interacting collection of binary valued neurons. The applications of machine learning are numerous and include, for example, finance, biotechnology, e-commerce, chemistry, insurance and security. In particular, neural networks have been successfully used in finance for portfolio analysis~\cite{barr1998predictive} and credit approval~\cite{norris1999system}, as well as in e-commerce for user ratings of online stores~\cite{lasa2007system}.

The Hebbian approach is the most natural learning method for neural networks that are fully visible and with undirected connections~\cite{hebb1949organization}. Hebbian learning specifies the connection strength between neurons according to the number of times that they fire together within training data. The output of Hebbian learning is a real symmetric weighting matrix with zero diagonal, whose elements correspond to the connection weights between neuron pairs. This matrix is then used as a component within machine learning algorithms. Undirected and fully visible neural networks, such as the canonical Hopfield network~\cite{hopfield1982neural}, are often utilized as an associative memory for pattern recognition, as well as for optimization problems such as the traveling salesman problem~\cite{mackay2003information}.

Quantum machine learning (QML) combines parts of quantum mechanics, such as quantum algorithms, with machine learning~\cite{biamonte2016quantum}. It can be split into two broad categories based upon whether the input data is of a classical nature or quantum. In classical-input QML, the data is initially classical and must be converted into quantum states for processing by a quantum device. This approach promises potentially exponential increases in data handling capacity and processing speed, as well as feasible circuit size scaling. However, such advantages hinge on the fundamental question of whether classical data can be converted efficiently into quantum data~\cite{aaronson2015read,giovannetti2008quantum,soklakov2006efficient}. Alternatively, quantum-input QML assumes that the relevant data is already of a quantum nature, taking the perspective that the QML algorithm is an element in the toolchain of a composite quantum device. Although perhaps less studied at present, quantum-input QML holds great potential given the predicted development of quantum devices over the coming years.

Nevertheless, QML algorithms can typically function in both the classical and quantum input regimes. A variety of new approaches has already been discovered, including for anomaly detection~\cite{liu2017quantum}, data fitting~\cite{wiebe2012quantum}, and support vector machines~\cite{rebentrost2014quantum}. Recent work on quantum neural networks~\cite{schuld2014quest} 
include the use of amplitude amplification~\cite{wiebe2014quantum}, 
quantum annealing~\cite{amin2016quantum}, Helmholtz machines~\cite{benedetti2017quantum}, the alternating operator ansatz~\cite{verdon2017}, and parameterized 
unitaries \cite{farhi2018}.

Learning makes up an important stage of any QML algorithm. In this work, we detail a new controlled method of quantum state exponentiation (QSE) using batches of quantum data and show how this method can form the quantum analogue of Hebbian learning. Here, a quantum Hebbian learning (QHL) device is composed of a series of processing qubits and an ancilla learning qubit, with training data inputted from a register of data qubits. By repetitively performing partial swaps between the processing and data qubits, with control on the learning qubit, a QSE of a mixed state representing the Hebbian learning matrix can be enacted on the processing qubits. We begin by outlining our approach, called batched controlled QSE (bcQSE), in Sec.~\ref{Sec:qDME}. We next breakdown the controlled partial swap operation into standard one and two qubit gates from the Clifford+$T$ set in Sec.~\ref{Sec:GateCount}, allowing us to perform a gate count as well as an analysis of errors and efficiency of bcQSE. QHL and its realization through bcQSE is formalized in Sec.~\ref{Sec:qHeb}. Our findings represent important precursor steps in the development of a concrete QHL device. We then conclude in Sec.~\ref{sec:con} with a discussion.

\section{Quantum state exponentiation}\label{Sec:qDME}

\subsection{Established methods}

We first outline the established results on QSE. Given any (possibly unknown) quantum state $\rho$, one can transform another system according to the unitary $e^{- i t \rho}$ for some time $t$ using the protocol outlined in the following theorem.
\begin{theorem}\label{theoremDMExp} { \rm (QSE)}
Let $\rho$ be a density matrix of $N$ qubits. 
By partial swapping an $N$ qubit quantum system with a single copy of $\rho$, 
one can enact $e^{- i  \Delta t \rho}$ on the system to error $\Ord{\Delta t^2}$ in diamond norm. By repetitively partial swapping an $N$ qubit quantum system with a number of copies $n \geq 1$ of $\rho$ that scales as $n = \Ord{t^2/\epsilon}$, one can enact $e^{- i  t \rho}$ on the system to error $\epsilon>0$ in diamond norm. This protocol is optimal in $n$ in terms of the $t$ and $\epsilon$.
 \end{theorem}
This theorem is proved in Refs.~\cite{lloyd2013quantum,kimmel2017hamiltonian,marvian2016}. Here, the partial swap operation is $e^{-i \Delta t \mathcal{S}}$ for some time $\Delta t = t/n$, with $\mathcal{S}$ the swap operation. A pure-state version of QSE was highlighted in the context of emulating unitaries in Ref.~\cite{marvian2016}.

Importantly, QSE can be extended to a controlled version dependent upon an ancillary control qubit~\cite{lloyd2013quantum,kimmel2017hamiltonian,marvian2016}.
\begin{theorem}\label{theoremDMExp2} { \rm (Controlled QSE)}
 Let $\rho$ be a density matrix over $N$ qubits.  
By partial swapping an $N$ qubit quantum system with a single copy of $\rho$, 
one can enact  $e^{- i t  \ket{1}\bra{1} \otimes \rho }$ on the system to error $\Ord{\Delta t^2}$ in diamond norm. By repetitively partial controlled swapping an $N$ qubit quantum system with a number of copies $n \geq 1$ of $\rho$ that scales as $n = \Ord{t^2/\epsilon}$, one can enact $e^{- i t  \ket{1}\bra{1} \otimes \rho }$ on the system to error $\epsilon$ in diamond norm. This protocol is optimal in $n$ in terms of the $t$ and $\epsilon$.
\end{theorem}
See Refs. \cite{lloyd2013quantum,kimmel2017hamiltonian,marvian2016} for the proof. Here, $e^{-i \Delta t\ket{1}\bra{1} \otimes\mathcal{S}}$ is the partial controlled swap operation.

\subsection{Batching}\label{Sec:qDMEBatch}

We now specify our batched approach to QSE within the context of quantum data. Consider a batch of $M$ pieces of $N$ qubit quantum data, each represented by a $d=2^{N}$-dimensional pure quantum state $| x^{(m)} \rangle$. A statistical ensemble of this quantum data is given by the $N$ qubit mixed state
\begin{equation}\label{Eq:Ensemble}
\rho = \frac{1}{M}\sum_{m=1}^{M}| x^{(m)} \rangle\langle x^{(m)} |.
\end{equation}
We shall see in the following that $\rho$ is an important object, containing all the relevant information required to capture the Hebbian learning matrix. The following theorem shows that $e^{- i t \rho}$ can be enacted on a series of $N$ processing qubits by repetitively partial swapping multiple batches of quantum data. A batch partial swap constitutes $M$ partial swaps between the processing qubits and collections of data supplying qubits, each in the ordered quantum states $\{| x^{(m)} \rangle\}_{m=1}^{M}$. We focus on the situation where there is additional control on an ancilla qubit, which we refer to as the learning qubit due the link with QHL discussed in Sec.~\ref{Sec:qHeb}, and hence call the process batched controlled QSE (bcQSE).
\begin{theorem} { \rm (bcQSE)}\label{BatchedQSE}
Let $\{| x^{(m)} \rangle\}_{m=1}^{M}$ be a batch of $M$ quantum states, each quantum state being of $N$ qubits. By repetitively batch partial swapping an $N$ qubit quantum system with a number of batches $n\geq1$ of this quantum data that scales as $n= \Ord{t^2/\epsilon}$, one can enact $e^{- i t \ket{1}\bra{1} \otimes \rho}$ on the system to error $\epsilon$ in diamond norm.
\end{theorem}
\begin{proof}
We discuss the single step before obtaining multiple steps via repeating the single step for a single batch of $M$ states then subsequently for $n$ batches.  For a short time $\Delta t$, which later will be set to $\Delta t = t / n$, define the $M$ unitaries
\begin{equation}
\mathcal{U}_m := \ket 0 \bra 0 \otimes \mathbbm I + \ket 1 \bra 1 \otimes e^{-i \frac{ \Delta t}{M} \ket{ x^{(m)}} \bra {x^{(m)}} }.
\end{equation}
From Theorem~\ref{theoremDMExp2}, a transformation $\mathcal W_m$ on the processing qubits can be simulated that approximates $\mathcal{U}_m$ by partial swapping one copy of $\ket{ x^{(m)}}$ from data supplying qubits. The error is $\Ord{\Delta t^2/M^2}$  in diamond norm.

For a single batch, we perform $M$ partial swaps between processing and data supplying qubits, with the data supplying qubits sequentially prepared in the states $| x^{(1)} \rangle,| x^{(2)} \rangle,\ldots,| x^{(M)} \rangle$. This simulates the sequence of unitaries $\mathcal{U}_{M}\hdots \mathcal{U}_{1}$ approximately via the sequence $\mathcal{W}_{M}\hdots \mathcal{W}_{1}$. The error for simulating $\mathcal{U}_{M}\hdots \mathcal{U}_{1}$ compounds $M$ times the single step error and hence is $\Ord{\Delta t^2/M}$. 
A single batch takes $M$ controlled partial swaps (for a shorter time $\Delta t/M$) and $M$ qubit registers to encode all the pure training states.

Repeating the process for $n\geq 1$ batches simulates the sequence
$(\mathcal{U}_{M}\hdots \mathcal{U}_{1})^n$. The error compounds $n$ times to be $\Ord{n \Delta t^2 /M}$. Replacing $\Delta t = t/n$ leads to an error of $\Ord{t^2/nM }$. Our target is for the partial swap steps $(\mathcal{W}_{M}\hdots \mathcal{W}_{1})^n$ to approximate $e^{- i  t \ket{1}\bra{1} \otimes \rho}$.  
The error is 
 \begin{eqnarray}
\epsilon :&=& \left \Vert (\mathcal{W}_{M}\hdots \mathcal{W}_{1})^n - e^{- i  t \ket{1}\bra{1} \otimes \rho}\right \Vert_{\diamond} \nonumber \\
 &\leq& \left \Vert (\mathcal{W}_{M}\hdots \mathcal{W}_{1})^n -(\mathcal{U}_{M}\hdots \mathcal{U}_{1})^n \right \Vert_{\diamond} \nonumber \\
 &&+ \left \Vert (\mathcal{U}_{M}\hdots \mathcal{U}_{1})^n - e^{- i  t \ket{1}\bra{1} \otimes \rho}\right \Vert_{\diamond} \nonumber \\
 &=& \Ord{\frac{t^2}{nM}} + \Ord{\frac{t^2}{n}} = \Ord{\frac{t^2}{n}}.
\end{eqnarray}
Here we have used the triangle inequality and, in the last step, the Lie product formula  \cite{Childs2017} to approximate $e^{- i  t \ket{1}\bra{1} \otimes \rho}$ with $(\mathcal{U}_{M}\hdots \mathcal{U}_{1})^n$ with the error $\Ord{ \frac{ t^{2}}{n} }$.
The overall error compounds the two sources and is hence $\epsilon = \Ord{t^2/n }$.
 
\end{proof}

Figure~\ref{fig2} illustrates the bcQSE protocol. It requires $nM$ collections of $N$ data supplying qubits that are supplied in $n$ batches of the $M$ sequential data states. This approach is preferable when the user has access to multiple copies of the quantum data in pure state form, rather than the statistical ensemble.  
In the next section, we break down the controlled partial swap into standard one and two qubit gates from the Clifford+$T$ set~\cite{nielsen2002quantum}. Being the elementary transformation of both batched and non-batched QSE, this decomposition makes more concrete the implementation of QSE. Focusing on bcQSE, we then perform an analysis of errors and efficiency.

\begin{figure}
\begin{flushleft}
{\noindent \bf (a)}
\end{flushleft}
\vspace{-0.8cm}\hspace{0.5cm}
\begin{minipage}{0.5 \textwidth}
$$
\Qcircuit @C=1em @R=.1em {
& \lstick{\mbox{Learning qubit    }} & \ctrl{33} & \qw & \\
& & & & \\
& & & & \\
& & & & \\
& & & & \\
& & & & \\
& & & & \\
& & & & \\
& & & & \\
& & & & \\
& & & & \\
& & & & \\
& & & & \\
& & & & \\
& & & & \\
& & & & \\
& & & & \\
& & & & \\
& & & & \\
& & & & \\
& & & & \\
& & & & \\
& & & & \\
& & & & \\
& & & & \\
& & & & \\
& & & & \\
& & & & \\
& & & & \\
& & & & \\
& & & & \\
& & & & \\
& & & & \\
& & \multigate{7}{e^{-i t \rho}} & \qw & \\
& & \ghost{e^{-i t \rho}} & \qw & \\
& & \ghost{e^{-i t \rho}} & \qw & \\
& \lstick{\mbox{Processing    }} & \ghost{e^{-i t \rho}} & \qw & \\
& \lstick{\mbox{qubits    }} & \ghost{e^{-i t \rho}} & \qw & \\
& & \ghost{e^{-i t \rho}} & \qw & \\
& & \ghost{e^{-i t \rho}} & \qw & \\
& & \ghost{e^{-i t \rho}} & \qw \gategroup{33}{2}{41}{4}{1.2em}{..} & \\
}
$$
\end{minipage}
\vspace{-0.4cm}
\begin{flushleft}
{\noindent \bf (b)}
\end{flushleft}
\vspace{-0.6cm}
\begin{minipage}{0.05\textwidth}
\empty
\end{minipage}
\begin{minipage}{0.25\textwidth}
    $$
\Qcircuit @C=1em @R=.1em {
& \lstick{\mbox{Learning qubit}} & \ctrl{16} & \qw & \qw & \qw & \ctrl{16} & \qw & \qw & \qw & \ctrl{16} & \qw & \qw & \qw & \ctrl{16} & \qw &  \\
& & & & & & & & & & & & & & & & \\
& & & & & & & & & & & & & & & & \\
& & & & & & & & & & & & & & & & \\
& & & & & & & & & & & & & & & & \\
& & & & & & & & & & & & & & & & \\
& & & & & & & & & & & & & & & & \\
& & & & & & & & & & & & & & & & \\
& & & & & & & & & & & & & & & & \\
& & & & & & & & & & & & & & & & \\
& & & & & & & & & & & & & & & & \\
& & & & & & & & & & & & & & & & \\
& & & & & & & & & & & & & & & & \\
& & & & & & & & & & & & & & & & \\
& & & & & & & & & & & & & & & & \\
& & & & & & & & & & & & & & & & \\
& & \multigate{7}{ } & \qw & \ldots &  & \multigate{7}{ } & \qw & \ldots & & \multigate{7}{ } & \qw & \ldots &  & \multigate{7}{ } & \qw & \\
& & \ghost{ } & \qw & \ldots &  & \ghost{ } & \qw & \ldots & & \ghost{ } & \qw & \ldots &  & \ghost{ } & \qw & \\
& & \ghost{ } & \qw & \ldots &  & \ghost{ } & \qw & \ldots & & \ghost{ } & \qw & \ldots &  & \ghost{ } & \qw & \\
& \lstick{\mbox{Processing  }} & \ghost{ } & \qw & \ldots &  & \ghost{ } & \qw & \ldots & & \ghost{ } & \qw & \ldots &  & \ghost{ } & \qw & \\
& \lstick{\mbox{qubits  }} & \ghost{ } & \qw & \ldots &  & \ghost{ } & \qw & \ldots & & \ghost{ } & \qw & \ldots &  & \ghost{ } & \qw & \\
& & \ghost{ } & \qw & \ldots &  & \ghost{ } & \qw & \ldots & & \ghost{ } & \qw & \ldots &  & \ghost{ } & \qw & \\
& & \ghost{ } & \qw & \ldots &  & \ghost{ } & \qw & \ldots & & \ghost{ } & \qw & \ldots &  & \ghost{ } & \qw & \\
& & \ghost{ } \qwx[16] & \qw & \ldots &  & \ghost{ } \qwx[39] & \qw & \ldots & & \ghost{ }\qwx[62] & \qw & \ldots &  & \ghost{ }\qwx[85] & \qw \gategroup{17}{2}{24}{16}{1.2em}{..} & \\
& & & & & & & & & & & & & & & & \\
& & & & & & & & & & & & & & & & \\
& & & & & & & & & & & & & & & & \\
& & & & & & & & & & & & & & & & \\
& & & & & & & & & & & & & & & & \\
& & & & & & & & & & & & & & & & \\
& & & & & & & & & & & & & & & & \\
& & & & & & & & & & & & & & & & \\
& & & & & & & & & & & & & & & & \\
& & & & & & & & & & & & & & & & \\
& & & & & & & & & & & & & & & & \\
& & & & & & & & & & & & & & & & \\
& & & & & & & & & & & & & & & & \\
& & & & & & & & & & & & & & & & \\
& & & & & & & & & & & & & & & & \\
& & \multigate{7}{ } & \qw & \qw & \qw & \qw & \qw &  &  &  &  &  &  &  &  & \\
& & \ghost{ } & \qw & \qw & \qw & \qw & \qw &  &  &  &  &  &  &  &  & \\
& & \ghost{ } & \qw & \qw & \qw & \qw & \qw &  &  &  &  &  &  &  &  & \\
& \lstick{| x^{(1)} \rangle\,\,} & \ghost{ } & \qw & \qw & \qw & \qw & \qw &  &  &  &  &  &  &  &  & \\
& & \ghost{ } & \qw & \qw & \qw & \qw & \qw &  &  &  &  &  &  &  &  & \\
& & \ghost{ } & \qw & \qw & \qw & \qw & \qw &  &  &  &  &  &  &  &  & \\
& & \ghost{ } & \qw & \qw & \qw & \qw & \qw &  &  &  &  &  &  &  &  & \\
& & \ghost{ } & \qw & \qw & \qw & \qw & \qw &  &  &  &  &  &  &  &  & \\
& & & & & & & & & & & & & & & & \\
& & & & & & & & & & & & & & & & \\
& & & & & & & & & & & & & & & & \\
& & & & & & & & & & & & & & & & \\
& & & & & & & & & & & & & & & & \\
& & & & & & & & & & & & & & & & \\
& & & & \ddots & & & & & & & & & & & & \\
& & & & & & & & & & & & & & & & \\
& & & & & & & & & & & & & & & & \\
& & & & & & & & & & & & & & & & \\
& & & & & & & & & & & & & & & & \\
& & & & & & & & & & & & & & & & \\
& & & & & & & & & & & & & & & & \\
& & & & & & & & & & & & & & & & \\
& & & & & & & & & & & & & & & & \\
& & \qw & \qw & \qw & \qw & \multigate{7}{ } & \qw &  &  &  &  &  &  &  &  & \\
& & \qw & \qw & \qw & \qw & \ghost{ } & \qw &  &  &  &  &  &  &  &  & \\
& & \qw & \qw & \qw & \qw & \ghost{ } & \qw &  &  &  &  &  &  &  &  & \\
& \lstick{| x^{(M)} \rangle\,\,} & \qw & \qw & \qw & \qw & \ghost{ } & \qw &  &  &  &  &  &  &  &  & \\
& & \qw & \qw & \qw & \qw & \ghost{ } & \qw &  &  &  &  &  &  &  &  & \\
& & \qw & \qw & \qw & \qw & \ghost{ } & \qw &  &  &  &  &  &  &  &  & \\
& & \qw & \qw & \qw & \qw & \ghost{ } & \qw &  &  &  &  &  &  &  &  & \\
& & \qw & \qw & \qw & \qw & \ghost{ } \gategroup{40}{2}{70}{8}{1.2em}{..} & \qw &  &  &  &  &  &  &  &  & \\
& & & & & & & & & & & & & & & & \\
& & & & & & & & & & & & & & & & \\
& & & & & & & & & & & & & & & & \\
& & & & & & & & & & & & & & & & \\
& & & & & & & & & & & & & & & & \\
& & & & & & & & & & & & & & & & \\
& & & & \mbox{First batch} & &  & & \ddots & & & & & & & & \\
& & & & & & & & & & & & & & & & \\
& & & & & & & & & & & & & & & & \\
& & & & & & & & & & & & & & & & \\
& & & & & & & & & & & & & & & & \\
& & & & & & & & & & & & & & & & \\
& & & & & & & & & & & & & & & & \\
& & & & & & & & & & & & & & & & \\
& & & & & & & & & & & & & & & & \\
& & & & & & & &  &  & \multigate{7}{ } & \qw & \qw & \qw & \qw & \qw & \\
& & & & & & & &  &  & \ghost{ } & \qw & \qw & \qw & \qw & \qw & \\
& & & & & & & &  &  & \ghost{ } & \qw & \qw & \qw & \qw & \qw & \\
& & & & & & & &  & \lstick{| x^{(1)} \rangle\,\,}  & \ghost{ } & \qw & \qw & \qw & \qw & \qw & \\
& & & & & & & &  &  & \ghost{ } & \qw & \qw & \qw & \qw & \qw & \\
& & & & & & & &  &  & \ghost{ } & \qw & \qw & \qw & \qw & \qw & \\
& & & & & & & &  &  & \ghost{ } & \qw & \qw & \qw & \qw & \qw & \\
& & & & & & & &  &  & \ghost{ } & \qw & \qw & \qw & \qw & \qw & \\
& & & & & & & & & & & & & & & & \\
& & & & & & & & & & & & & & & & \\
& & & & & & & & & & & & & & & & \\
& & & & & & & & & & & & & & & & \\
& & & & & & & & & & & & & & & & \\
& & & & & & & & & & & & & & & & \\
& & & &  & &  & &   & & & &  \ddots & & & & \\
& & & & & & & & & & & & & & & & \\
& & & & & & & & & & & & & & & & \\
& & & & & & & & & & & & & & & & \\
& & & & & & & & & & & & & & & & \\
& & & & & & & & & & & & & & & & \\
& & & & & & & & & & & & & & & & \\
& & & & & & & & & & & & & & & & \\
& & & & & & & & & & & & & & & & \\
& & & & & & & &  &  & \qw & \qw & \qw & \qw & \multigate{7}{ } & \qw & \\
& & & & & & & &  &  & \qw & \qw & \qw & \qw & \ghost{ } & \qw & \\
& & & & & & & &  &  & \qw & \qw & \qw & \qw & \ghost{ } & \qw & \\
& & & & & & & &  & \lstick{| x^{(M)} \rangle\,\,} & \qw & \qw & \qw & \qw & \ghost{ } & \qw & \\
& & & & & & & &  &  & \qw & \qw & \qw & \qw & \ghost{ } & \qw & \\
& & & & & & & &  &  & \qw & \qw & \qw & \qw & \ghost{ } & \qw & \\
& & & & & & & &  &  & \qw & \qw & \qw & \qw & \ghost{ } & \qw & \\
& & & & & & & &  &  & \qw & \qw & \qw & \qw & \ghost{ } \gategroup{86}{10}{116}{16}{1.2em}{..} & \qw & \\
& & & & & & & & & & & & & & & & \\
& & & & & & & & & & & & & & & & \\
& & & & & & & & & & & & & & & & \\
& & & & & & & & & & & & & & & & \\
& & & & & & & & & & & & & & & & \\
& & & & & & & & & & & & & & & & \\
& & & &  & &  & &  & & & & \mbox{$n$-th batch} & & & & \\
}
$$
\end{minipage}
\begin{minipage}{0.25\textwidth}
\begin{flushleft}
\vspace{-5cm}
\hspace{-1.6cm}\includegraphics[width=0.6\textwidth]{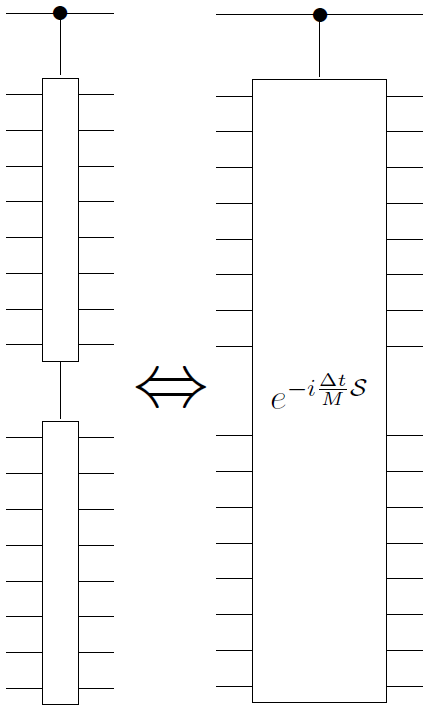}
\\ \hspace{-1.0cm} { \bf Key}
\end{flushleft}
\end{minipage}
\vspace{-0.2cm}
    \caption{(a) Quantum state exponentiation (QSE) of a state $\rho$ onto a series of processing qubits, with control from an ancilla learning qubit. (b) Batched controlled quantum state exponentiation (bcQSE) of a batch of quantum states.
    Quantum data is organized into $n$ batches of $M$ pure states. The first collection of $N$ qubits is prepared in $|x^{(1)} \rangle$, then the next collection in $|x^{(2)} \rangle$, and so forth until the $M$-th collection which is in $|x^{(M)} \rangle$. This represents the first batch, and the process is repeated $n$ times. 
}
    \label{fig2}
\end{figure}

\section{Gate decomposition}\label{Sec:GateCount}

The Clifford set of gates consists of the Hadamard gate $H$, the phase gate $S$, the global phase gate $W$, and the CNOT gate. Any quantum circuit consisting of only these gates can be simulated efficiently on a classical computer~\cite{nielsen2002quantum}. The conventional approach to extending beyond classical simulability is to add the $\pi/8$ gate $T$. The matrix representations of these gates with respect to the computational basis are:
\begin{eqnarray}
& & H = \frac{1}{\sqrt{2}} \left(\begin{array}{cc}
1 & 1 \\
1 & -1
\end{array} \right),
\qquad
S = \left(\begin{array}{cc}
1 & 0 \\
0 & i
\end{array} \right),
\nonumber \\
& & T = \left(\begin{array}{cc}
1 & 0 \\
0 & e^{i \frac{\pi}{4}}
\end{array} \right),
\qquad
W = \left(\begin{array}{cc}
e^{i \frac{\pi}{4}} & 0 \\
0 & e^{i \frac{\pi}{4}}
\end{array} \right),
\nonumber \\
& & CNOT = \left(\begin{array}{cccc}
1 & 0 & 0 & 0 \\
0 & 1 & 0 & 0 \\
0 & 0 & 0 & 1 \\
0 & 0 & 1 & 0 \\
\end{array}\right). \nonumber
\end{eqnarray}
The ability to perform Clifford+$T$ gates is established for a variety of physical implementations of quantum circuits, such as those using superconducting qubits, trapped ions or nuclear magnetic resonance. 

\begin{figure*}
\hspace{-5.5cm}
\begin{minipage}{.43\textwidth}
{\hspace{0.4cm}\bf (a)}
\vspace*{-0.3cm}
$$
\Qcircuit @C=1em @R=.1em {
& \lstick{\mbox{Learning qubit    }} & \ctrl{33} & \qw \\
& & & \\
& & & \\
& & & \\
& & & \\
& & & \\
& & & \\
& & & \\
& & & \\
& & & \\
& & & \\
& & & \\
& & & \\
& & & \\
& & & \\
& & & \\
& & & \\
& & & \\
& & & \\
& & & \\
& & & \\
& & & \\
& & & \\
& & & \\
& & & \\
& & & \\
& & & \\
& & & \\
& & & \\
& & & \\
& & & \\
& & & \\
& & & \\
& \lstick{a_{1}} & \multigate{47}{e^{-i \frac{\Delta t}{M} \mathcal{S}}} & \qw \\
& \lstick{a_{2}} & \ghost{e^{-i \frac{\Delta t}{M} \mathcal{S}}} & \qw \\
& & \ghost{e^{-i \frac{\Delta t}{M} \mathcal{S}}} & \qw \\
& & \ghost{e^{-i \frac{\Delta t}{M} \mathcal{S}}} & \qw \\
& \lstick{\vdots} & \ghost{e^{-i \frac{\Delta t}{M} \mathcal{S}}} & \qw \\
& & \ghost{e^{-i \frac{\Delta t}{M} \mathcal{S}}} & \qw \\
& & \ghost{e^{-i \frac{\Delta t}{M} \mathcal{S}}} & \qw \\
& \lstick{a_{N}} & \ghost{e^{-i \frac{\Delta t}{M} \mathcal{S}}} & \qw \\
& & & \\
& & & \\
& & & \\
& & & \\
& & & \\
& & & \\
& & & \\
& & & \\
& & & \\
& & & \\
& & & \\
& & & \\
& & & \\
& & & \\
& & & \\
& & & \\
& & & \\
& & & \\
& & & \\
& & & \\
& & & \\
& & & \\
& & & \\
& & & \\
& & & \\
& & & \\
& & & \\
& & & \\
& & & \\
& & & \\
& & & \\
& & & \\
& \lstick{b_{1}} & \ghost{e^{-i \frac{\Delta t}{M} \mathcal{S}}} & \qw \\
& \lstick{b_{2}} & \ghost{e^{-i \frac{\Delta t}{M} \mathcal{S}}} & \qw \\
& & \ghost{e^{-i \frac{\Delta t}{M} \mathcal{S}}} & \qw \\
& & \ghost{e^{-i \frac{\Delta t}{M} \mathcal{S}}} & \qw \\
& \lstick{\vdots} & \ghost{e^{-i \frac{\Delta t}{M} \mathcal{S}}} & \qw \\
& & \ghost{e^{-i \frac{\Delta t}{M} \mathcal{S}}} & \qw \\
& & \ghost{e^{-i \frac{\Delta t}{M} \mathcal{S}}} & \qw \\
& \lstick{b_{N}} & \ghost{e^{-i \frac{\Delta t}{M} \mathcal{S}}} & \qw \\
}
$$
\end{minipage}
\begin{minipage}{.32\textwidth}
\vspace{-0.57cm}
{\hspace{5.0cm} \bf (b)}
\vspace{-0.3cm}
$$
\Qcircuit @C=1em @R=1em {
& \lstick{\mbox{Learning qubit}}& \qw & \qw & \qw & \qw & \qw & \qw & \ctrl{12} & \qw & \qw & \qw & \qw & \qw & \qw & \qw \\
& & & & & & & & & & & & & & & \\
& & & & & & & & & & & & & & & \\
& \lstick{a_{1}} & \targ & \ctrl{1} & \ctrl{1} & \qw & \qw & \qw & \qw & \qw & \qw & \qw & \ctrl{1} & \ctrl{1} & \targ & \qw  \\
& \lstick{b_{1}}  & \ctrl{-1} & \gate{e^{-i \frac{\pi}{4} Y}} & \ctrl{8} & \qw & \qw & \qw & \qw & \qw & \qw & \qw & \ctrl{8} & \gate{e^{i \frac{\pi}{4} Y}} & \ctrl{-1} & \qw \\
& \lstick{a_{2}}  & \targ & \ctrl{1} & \qw & \ctrl{1} & \qw & \qw & \qw & \qw & \qw & \ctrl{1} & \qw & \ctrl{1} & \targ & \qw \\
& \lstick{b_{2}}  & \ctrl{-1} & \gate{e^{-i \frac{\pi}{4} Y}} & \qw & \ctrl{6} & \qw & \qw & \qw & \qw & \qw & \ctrl{6} & \qw & \gate{e^{i \frac{\pi}{4} Y}} & \ctrl{-1} & \qw \\
& & \vdots & \vdots & & & \ddots & & & & \iddots & & & \vdots & \vdots \\ \\
& \lstick{a_{N}} & \targ & \ctrl{1} & \qw & \qw & \qw & \ctrl{1} & \qw & \ctrl{1} & \qw & \qw & \qw & \ctrl{1} & \targ & \qw \\
& \lstick{b_{N}} & \ctrl{-1} & \gate{e^{-i \frac{\pi}{4} Y}} & \qw & \qw & \qw & \ctrl{2} & \qw & \ctrl{2} & \qw & \qw & \qw & \gate{e^{i \frac{\pi}{4} Y}} & \ctrl{-1} & \qw \\
& & & & & & & & & & & & & & & \\
& \lstick{\ket{0}} & \qw & \qw & \targ & \targ & \qw & \targ & \gate{e^{- i \frac{\Delta t}{M} Z}} & \targ & \qw & \targ & \targ & \qw & \qw & \qw \\
}
$$
\end{minipage}
\caption{The controlled partial swap operation (a) between the $N$ qubit collection $a_{1},a_{2},\ldots,a_{N}$ and the $N$ qubit collection $b_{1},b_{2},\ldots,b_{N}$, with learning qubit control, can be decomposed into two and three qubit quantum gates according to (b), following steps from Ref.~\cite{childs2003exponential}. Here, an additional ancilla qubit initialized in the state $\ket{0}$ is utilized to perform the transformation. The symbols used follow standard conventions~\cite{nielsen2002quantum}.    
}
    \label{fig3}
\end{figure*}
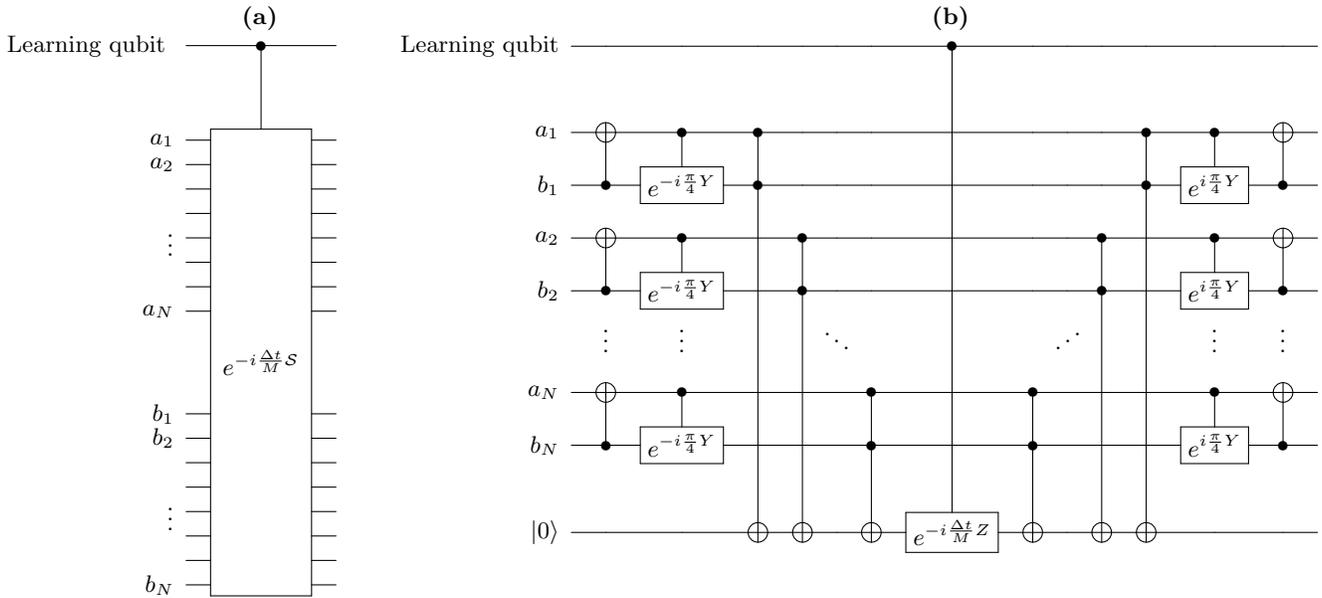

The controlled partial swap operation used in Fig.~\ref{fig2} between $N$ processing qubits and $N$ data supplying qubits for a time $\frac{\Delta t}{M}$, with control on the learning qubit, can be first decomposed into two and three qubit gates following similar steps to Ref.~\cite{childs2003exponential}. The result is presented in Fig.~\ref{fig3}. Here, we label the $N$ processing qubits as $a_{1},a_{2},\ldots,a_{N}$ and the $N$ data supplying qubits as $b_{1},b_{2},\ldots,b_{N}$ (although the choice of labeling is symmetric). The next step is to further decompose these two and three qubit gates into gates from the Clifford+$T$ set. In Fig.~\ref{fig4}, we show such decompositions for (a) the controlled $e^{-i \frac{\pi}{4}Y}$ unitary, (b) the Toffoli gate, and (c) the time-dependent controlled $e^{- i \frac{\Delta t}{M} Z}$ unitary, where $Y$ and $Z$ are qubit Pauli gates. In Fig.~\ref{fig4}~(c), we show the controlled $e^{- i \frac{\Delta t}{M} Z}$ as a composition of two CNOT gates and the two single qubit unitaries $e^{\pm i \frac{\Delta t}{2M} Z}$, which can be approximately decomposed into Clifford+$T$ gates using a technique given in Ref.~\cite{ross2014optimal}, see the following for further details.

\begin{figure*}
\begin{flushleft}
\noindent {\bf (a)}
\hspace{1cm}
\begin{minipage}{0.09\textwidth}
\begin{center}
$\Qcircuit @C=1em @R=1.4em {
& \ctrl{1} & \qw \\
& \gate{e^{-i \frac{\pi}{4} Y }} & \qw
}$
\end{center}
\end{minipage}
\begin{minipage}{0.06\textwidth}
\begin{center}
\vspace{0.1cm}
$\Leftrightarrow$
\end{center}
\end{minipage}
\begin{minipage}{0.7\textwidth}
\begin{flushleft}
\vspace{0.25cm}
$\Qcircuit @C=1em @R=1.4em {
& \qw & \qw & \qw & \qw & \qw & \ctrl{1} & \qw & \qw & \qw & \qw & \qw & \ctrl{1} & \qw  \\
& \gate{S} & \gate{H} & \gate{T} & \gate{H} & \gate{S^{\dagger}} & \targ & \gate{S} & \gate{H} & \gate{T^{\dagger}} & \gate{H} & \gate{S^{\dagger}} &  \targ & \qw
}$
\end{flushleft}
\end{minipage}
\vspace{1cm}
\\ \noindent {\bf (b)}
\hspace{1cm}
\begin{minipage}{0.09\textwidth}
\begin{center}
$\Qcircuit @C=1em @R=1.7em {
& \ctrl{2} & \qw \\
& \ctrl{1} & \qw \\
& \targ & \qw
}$
\end{center}
\end{minipage}
\begin{minipage}{0.06\textwidth}
\begin{center}
$\Leftrightarrow$
\end{center}
\end{minipage}
\begin{minipage}{0.7\textwidth}
\vspace{-.8mm}\hspace{-2.6cm}
$\Qcircuit @C=1em @R=.69em {
& \qw & \qw & \qw & \ctrl{2} & \qw & \qw & \qw & \ctrl{2} & \qw & \ctrl{1} & \qw & \ctrl{1} & \gate{T} & \qw \\
& \qw & \ctrl{1} & \qw & \qw & \qw & \ctrl{1} & \qw & \qw & \gate{T^{\dagger}} & \targ & \gate{T^{\dagger}} & \targ & \gate{S} & \qw \\
& \gate{H} & \targ & \gate{T^{\dagger}} & \targ & \gate{T} & \targ & \gate{T^{\dagger}} & \targ & \gate{T} & \gate{H} & \qw & \qw & \qw & \qw
}$
\end{minipage}
\vspace{1cm}
\\ \noindent {\bf (c)}
\hspace{0.9cm}
\begin{minipage}{0.10\textwidth}
\begin{center}
\vspace{-0.2cm}
$\Qcircuit @C=1em @R=1.4em {
& \ctrl{1} & \qw \\
& \gate{e^{- i \frac{\Delta t}{M} Z }} & \qw
}$
\end{center}
\end{minipage}
\begin{minipage}{0.06\textwidth}
\begin{center}
$\Leftrightarrow$
\end{center}
\end{minipage}
\begin{minipage}{0.6\textwidth}
\hspace{-6.1cm}
$\Qcircuit @C=1em @R=1.4em {
& \qw & \ctrl{1} & \qw & \ctrl{1} & \qw  \\
& \gate{e^{ - i \frac{\Delta t}{2M} Z}} & \targ & \gate{e^{ i \frac{\Delta t}{2M} Z}} &  \targ & \qw
}$
\end{minipage}
\hspace{1cm}

\end{flushleft}
    \caption{A decomposition of the two and three qubit gates in Fig.~\ref{fig3}~(b) into standard one and two qubit gates from the Clifford+$T$ set for (a) the controlled unitary $e^{-i \frac{\pi}{4}Y}$ and (b) the Toffoli gate, as given in~\cite{nielsen2002quantum}. The inverse gates $T^{\dagger}$ and $S^{\dagger}$ are assumed to be part of the Clifford+$T$ set and that $e^{i \frac{\pi}{4}Y}$ can be realized by simply swapping the positions of $T$ and $T^{\dagger}$ in (a). 
In (c), the time-dependent controlled unitary $e^{- i \frac{\Delta t}{M}Z}$ is decomposed into two CNOT gates and two single qubit unitaries $e^{\pm i \frac{\Delta t}{2M}Z}$. We discuss in the main text how these unitaries can be approximately decomposed into Clifford+$T$ gates~\cite{ross2014optimal}.}
    \label{fig4}
\end{figure*}
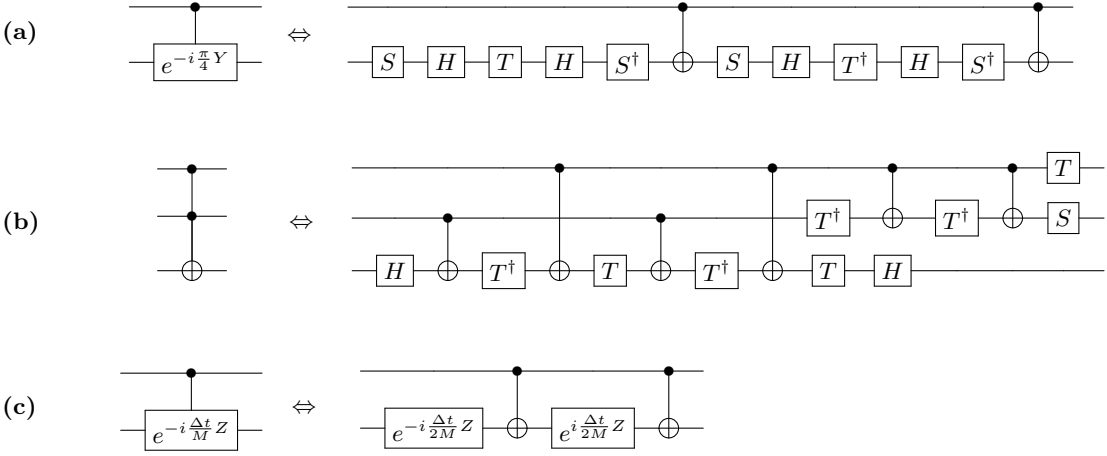

We now count the total number of Clifford+$T$ gates required to perform bcQSE. The gate count number is recorded in the vector $(H,S,W,CNOT,T)$. Herein, we assume that the inverse gates $T^{\dagger}$ and $S^{\dagger}$ have the same overhead as the corresponding $T$ and $S$ gates. Let us first consider the time-dependent qubit unitaries $e^{\pm i \frac{\Delta t}{2M}Z}$, which can be reduced using the following result~\cite{ross2014optimal}.
\begin{result}\label{Result1}
A unitary $e^{- i \tau Z}$ for any time $\tau$ can be approximated to error $\eta$ in operator norm using a probabilistically generated sequence of $H$, $S$, $W$, and $T$ gates with an approximate gate count number $(3 g_{\eta}, 2 g_{\eta}, g, 0, 3 g_{\eta})$, where $g_{\eta} :=\tOrd{ \log_{2} (1/\eta) }$ and $g \approx 10$.
\end{result}
This result can be confirmed numerically using an algorithm developed by Ross and Selinger~\cite{ross2014optimal}, who also showed rigorously that the number of required $T$ gates scales as $\tOrd{\log\left(\frac{1}{\eta}\right) }$, with $\tOrd{\cdot}$ indicating the presence of more slowly growing terms in $\eta$. 
Note that one can typically have a lower gate count requirement by selecting the best output from multiple runs of the algorithm. Hence, the two $e^{\pm i \frac{\Delta t}{2M}Z}$ gates can be performed to a combined error $2\eta$ using approximately $(6 g_{\eta}, 4 g_{\eta}, 2g, 0, 6 g_{\eta})$ gates.

This result can be combined with the decomposition in Figs.~\ref{fig3} and~\ref{fig4} to arrive at the following result.
\begin{theorem}\label{lemCtrlInfSwap} { \rm (Controlled partial swap gate count)}
Consider two $N$ qubit systems along with a single control qubit as well as an ancillary qubit prepared in $\ket{0}$. Assuming error-free Clifford+$T$ gates, there exists a quantum circuit performing the controlled partial swap operation $e^{- i \frac{\Delta t}{M} \mathcal{S}}$ to error $2 \eta$ in operator norm using a gate count of
\begin{equation} \label{eqGateCount}
\vec{g}(\eta) = (12N + 6 g_{\eta}, 10N + 4g_{\eta}, 2g, 18N + 2, 18N + 6g_{\eta}).
\end{equation}
\end{theorem}
\begin{proof}
The controlled partial swap in Fig.~\ref{fig3}~(a) is broken down into Clifford+$T$ gates in Fig.~\ref{fig3}~(b), which is adapted from Ref.~\cite{childs2003exponential}, and Fig.~\ref{fig4}.
The controlled unitary $e^{-i \frac{\pi}{4}Y}$ has a gate count of $(4,4,0,2,2)$ and is applied $2N$ times in Fig.~\ref{fig3}~(b) ($N$ of those times being the conjugate gate, with the same gate count). The Toffoli gate has a gate count of $(2,1,0,6,7)$ and is also applied $2N$ times. Next, the controlled $e^{- i \frac{\Delta t}{M}Z}$ unitary has a gate count of $(6g_{\eta},4g_{\eta},2g,2,6 g_{\eta})$ using Result~\ref{Result1} and is applied once. Finally, there are $2N$ additional CNOT gates in Fig.~\ref{fig3}~(b). Combining all of these elements results in a gate count of $\vec{g}(\eta) = (12N + 6 g_{\eta}, 10N + 4g_{\eta}, 2g, 18N + 2, 18N + 6g_{\eta})$. As all Clifford+$T$ gates are assumed to be error-free, the only source of error is in approximating the controlled $e^{- i \frac{\Delta t}{M}Z}$. This results in an overall error of $2\eta$.
\end{proof}

We can now perform a gate count of bcQSE and provide a detailed account of the errors arising from both the controlled partial swap and bcQSE. Errors can arise from two sources: (a) Clifford+$T$ gate errors and (b) unitary simulation errors. We denote the error (in operator norm) arising from a single application of each Clifford+$T$ gate by the vector $\vec{\epsilon}_{g} := (\epsilon_{H},\epsilon_{S},\epsilon_{W},\epsilon_{CNOT},\epsilon_{T})$. Theorem~\ref{lemCtrlInfSwap} then says that the controlled partial swap can be performed with imperfect gates to an error $\vec{\epsilon}_{g} \cdot \vec{g}(\eta) + 2 \eta$, combining errors both from the elementary gates and the simulations of $e^{\pm i \frac{\Delta t}{2M}Z}$. We now consider the overall error in bcQSE.

\begin{theorem}\label{Res:bcQSEg}
{\rm (bcQSE gate count and error)}
Consider $N$ processing qubits, a single learning qubit, $nM$ collections of $N$ data supplying qubits and $nM$ ancillary qubits in state $\ket{0}$. There exists a quantum circuit of controlled partial swap operations which performs bcQSE to error $2\eta$ in operator norm using a gate count of $nM \vec{g}(\eta)$.
Using imperfect Clifford+$T$ gates with errors $\vec{\epsilon}_{g} = (\epsilon_{H},\epsilon_{S},\epsilon_{W},\epsilon_{CNOT},\epsilon_{T})$, such bcQSE is performed to an error in the operator norm of
\begin{equation}\label{Eq:bcQSEError}
\epsilon=\alpha \frac{t^{2}}{n}+n M \left(\vec{\epsilon}_{g} \cdot \vec{g}(\eta)+ 2 \eta \right)
\end{equation}
for some constant $\alpha>0$.
\end{theorem}
\begin{proof}
This result holds since one needs to perform $nM$ controlled partial swaps to realize the bcQSE in Theorem~\ref{BatchedQSE}. 
A single application of the controlled partial swap operation can be carried out with an error in operator norm of  $\vec{\epsilon}_{g} \cdot \vec{g}(\eta)+ 2 \eta$. This operation must be repeated $nM$ times in bcQSE. Theorem~\ref{BatchedQSE} tells us that $\alpha\frac{t^{2}}{n}$ additional error (in operator norm, which is upper bounded by the diamond norm) arises due to approximating controlled QSE with the batched process, with $\alpha>0$ a constant factor determined by $\rho$.
\end{proof}

Performing bcQSE with imperfect Clifford+$T$ gates means that the error cannot be brought arbitrarily close to zero simply by increasing the number of batches $n$. Indeed, Eq.~(\ref{Eq:bcQSEError}) shows that the dominant error for large $n$ stems from the need to perform a large number of imperfect controlled partial swap operations. One can find an optimal regime by using a specific $n$ to minimize the resultant error. For a fixed choice of $\eta$, one can find the optimal $n$ via $\partial \epsilon/\partial n =0$, which results in
\begin{equation}\label{Eq:FixedError}
n_{\rm opt} = t\sqrt{\frac{\alpha }{M(\vec{\epsilon}_{g} \cdot \vec{g}(\eta)+2\eta)}}.
\end{equation}
with the constraint that $n_{\rm opt}\geq1$. If the gate error $M \left(\vec{\epsilon}_{g} \cdot \vec{g}(\eta)+ 2 \eta \right)$ is too large then the optimized $n_{\rm opt}$ is pushed below $1$ and the user must fix $n=1$. In this scenario, batched quantum state exponentiation fails. Practically, this constraint places limitations on the batch size $M$ and the individual gate errors for successful bcQSE.

Note that we find $n_{\rm opt}  = \Ord{t/ \sqrt{MN}}$ since $\alpha$, $\vec{\epsilon}_{g}$, and $\eta$ are constants and we have the lower bound
$\vec{g}(\eta)=\Omega\left(N\right)$ from Eq.~(\ref{eqGateCount}). As $N,M \geq 1$, we can conveniently use $n_ {\rm opt}  = \Ord{t}$, keeping in mind that $n_{\rm opt}\geq 1$. Using Eq.~(\ref{Eq:FixedError}) in Eq.~(\ref{Eq:bcQSEError}) results in the error $\epsilon = 2 t \sqrt{\alpha M (\vec{\epsilon}_{g} \cdot \vec{g}(\eta)+2\eta)}=\Ord{t \sqrt{M N}}$. This error is independent of the number of discretization steps taken and is the best one can do in the presence of constant gate errors. The user can also vary $\eta$ to choose a satisfactory compromise between the number of batches (which influences the gate count) and the overall error.

However, using quantum error correction the user may be free to change the Clifford+$T$ gate errors at a cost of increased physical resources, e.g.~by combining ancilla qubits and measurements with imperfect physical gates to perform a low error logical gate~\cite{nielsen2002quantum}. We can then represent the additional gate resources required to enact the Clifford+$T$ gates to an error $\vec{\epsilon}_{g}$  by the $5\times 5$ matrix $\bm{G}_{\rm ec}(\vec \epsilon_g)$, with the obtainable error being arbitrarily small provided that the physical gates have an error below a constant threshold. Each column of 
$\bm{G}_{\rm ec}(\vec \epsilon_g)$ is set to contain the physical gates required to perform one of the elementary logical gates to the corresponding error in $\vec{\epsilon}_{g}$. For example, the user can decrease a logical CNOT gate error $\epsilon_{CNOT}$ using a quantum error-correcting code that requires multiple noisy CNOTs along with other physical gates. 
The fourth column of $\bm{G}_{\rm ec}(\vec \epsilon_g)$ then determines the number of basic gates required for such a logical CNOT. 
We also denote the additional qubits required by the vector $\vec{q}_{\rm ec}(\vec \epsilon_g):= \left(q_{H}(\epsilon_{H}),q_{S}(\epsilon_{S}),q_{W}(\epsilon_{W}),q_{CNOT}(\epsilon_{CNOT}),q_{T}(\epsilon_{T})\right)$.

Recall that the error $2\eta$ associated with approximating the two time dependent gates $e^{\pm i \frac{\Delta t}{2M}Z}$ can be freely chosen. We know from Ross and Selinger~\cite{ross2014optimal} that the cost in terms of $T$ gates to perform a time-dependent $Z$ unitary to error $\eta$ is $\tOrd{\log(1/\eta)}$. Using quantum error correction methods, it is reasonable to assume that for the elementary gates themselves, every bit in precision also takes a linear cost to achieve, i.e.~$\bm{G}_{\rm ec}(\vec \epsilon_g)=\Ord{\log\left( \frac{1}{\vec{\epsilon}_{g}}\right)}$ and
 $\vec{q}_{\rm ec}(\vec \epsilon_g)=\Ord{\log\left( \frac{1}{\vec{\epsilon}_{g}}\right)}$, both elementwise~\cite{gottesman2009introduction}.
This setting allows for controlling the overall error $\epsilon$ of performing bcQSE. We now account for the qubit and gate count in this case.

\begin{theorem}\label{thmArbError}
{\rm (Performing bcQSE to arbitrary error)}
Let $\epsilon>0$ be the desired accuracy of bcQSE for $M$ quantum states per batch and $N$ qubits encoding every vector in a batch. Assume Clifford+$T$ gate errors $\vec{\epsilon}_{g}$, which take  additional resources in terms of gates 
$\bm{G}_{\rm ec}(\vec \epsilon_g)=\Ord{\log\left( \frac{1}{\vec{\epsilon}_{g}}\right)}$ and qubits
 $\vec{q}_{\rm ec}(\vec \epsilon_g)=\Ord{\log\left( \frac{1}{\vec{\epsilon}_{g}}\right)}$ to achieve 
 via quantum error correction. 
The number of batches required to perform bcQSE can be set to $n = \Ord{\frac{t^{2} + 1}{\epsilon}}$, where $n\geq 1$, requiring a number of physical gates and physical qubits of both
$\Ord{(N + \log (nM)) \log (nM (N+ \log nM ))} = \tOrd{N \log (nMN)}$ per single partial swap.
\end{theorem}
\begin{proof}
Take a constant $\delta' >0$ and set
$\eta = \frac{\delta'}{n^{2}M}$. From Theorem \ref{lemCtrlInfSwap} this implies that
$\vec{g} \left(\frac{\delta'}{n^{2} M}\right)=\Ord{N + \log (nM)}$. With another constant $\delta''>0$, take 
$\vec{\epsilon}_{g} = \delta''/\left (n^2 M \vec{g} \left(\frac{\delta'}{n^{2} M}\right)\right)$, where the inverse of $\vec g$ is defined element-wise. 
With $\vec 1$ the vector of ones, this implies that the gate count of a single partial swap is given by
$\vec 1^\intercal \bm{G}_{\rm ec}(\vec \epsilon_g) \vec{g} \left(\frac{\delta'}{n^{2} M}\right)=\Ord{\log\left( \frac{1}{\vec{\epsilon}_{g}}\right) (N + \log (nM)) }
=\Ord{(N + \log (nM)) \log (nM (N+ \log nM))}$ and the number of error correction qubits 
is given by $\vec{q}_{\rm ec}^{\ \intercal}(\vec \epsilon_g) \vec{g} \left(\frac{\delta'}{n^{2}M}\right)=\Ord{(N + \log (nM)) \log (nM (N+ \log nM))}$. 
With $\delta =  2\delta' + 5\delta''$ and $\alpha > 0$, we then have from Theorem \ref{Res:bcQSEg}
for the total error that $\epsilon = (\alpha t^2 + \delta)/n$, which implies $n = \Ord{\frac{t^{2} + 1}{\epsilon}}$.
\end{proof}

It is important to consider the number of physical qubits and the number of physical gates for any implementation of a quantum device, along with the resultant error. For bcQSE, one can consider how these quantities scale as a function of the number of qubits $N$ in a piece of quantum data, the number of quantum data states $M$, and the desired simulation time $t$. This can be achieved by first considering the number of batches $n$ required to perform bcQSE. We have discussed two regimes for $n$, for which always $n\geq 1$. They are, (i) when the overall error $\epsilon$ is partly determined by the fixed-error Clifford+$T$ gates, so that $n = \Ord{t}$ according to Eq.~\ref{Eq:FixedError}, and (ii) when the overall error $\epsilon$ can be determined by the user by altering the individual errors in the Clifford+$T$ gates via quantum error correction, so that $n = \Ord{\frac{t^{2}+1}{\epsilon} }$.

The total number of logical qubits required to perform bcQSE is $(nM+1)(N+1)=\Ord{nMN}$. Indeed, we require $N$ processing qubits and a single learning qubit, which receive $nM$ partial swaps requiring $nMN$ qubits along with $nM$ ancilla qubits in $| 0 \rangle$. On the other hand, we have seen in Theorem~\ref{Res:bcQSEg} that the total number of gates required for bcQSE is $nM \vec{g}(\eta)=\Ord{nMN}$.
In the error corrected setting, each logical qubit is replaced by multiple physical qubits and 
additional physical gates are required to implement the logical gates.
The scaling of the relevant quantities is summarized in Table~\ref{Table:Scaling}. In each case, scaling with $N$, $M$, and $t$, along with $\epsilon$ for regime (ii), is never worse than polynomial, indicating efficiency. 
A physical implementation can in principle handle exponentially large vectors without hitting a ``brick wall'' of efficiency. 
The second regime using quantum error correction is more costly than the first regime because the user must pay a qubit and gate cost to control the error.

\begin{table}
\caption{\label{Table:Scaling} Efficiency analysis of performing bcQSE using Clifford+$T$ gates for relevant figures of merit in terms of the number of processing qubits $N$, the number of data states $M$, the simulation time $t$, and the error $\epsilon$.}
\begin{ruledtabular}
\begin{tabular}{lcc}
 & \multicolumn{2}{c} {Gate error regime} \\
Quantity & (i) fixed & (ii) error corrected \\ \hline
Error & $\Ord{t \sqrt{MN}}$ & $\epsilon$ (user specified) \\
$n\geq 1$ & $\Ord{t}$ & $\Ord{\frac{t^{2}+1}{\epsilon}}$ \\
Qubit number & $\Ord{t MN}$ & $\tOrd{\frac{t^{2} +1}{\epsilon} M N  \log\left(\frac{t^{2}+1}{\epsilon}M N\right)}$  \\
Gate count & $\Ord{t MN}$ & $\tOrd{\frac{t^{2}+1}{\epsilon} M N \log\left(\frac{t^{2}+1}{\epsilon}M N\right)}$
\end{tabular}
\end{ruledtabular}
\end{table}

\section{Application to quantum Hebbian learning}\label{Sec:qHeb}

Controlled QSE allows a quantum state $\rho$ to be simulated on another system for a chosen time $t$, i.e. for application of the unitary transformation $e^{-i t \ket{1}\bra{1}\otimes \rho}$. This transformation can provide a way to use $\rho$ operationally within other quantum algorithms. For example, one can perform matrix multiplication with $\rho$~\cite{wiebe2014quantum}, or find its eigenvalues using quantum phase estimation~\cite{harrow2009quantum}. Our protocol for bcQSE is suited to the case where there is access to the pure quantum states composing the ensemble. Indeed, we show in the following how bcQSE can be thought of as a prototype quantum version of Hebbian learning.

\subsection{Classical Hebbian learning}

We first introduce the conventional approach to Hebbian learning to help motivate our following definition of quantum Hebbian Learning. Consider an artificial neural network consisting of $d$ binary-valued neurons $x_{i} \in \{1,-1\}$ with $i \in \{1, 2, \ldots , d\}$~\cite{mcculloch1943logical}, that are together described by the activation pattern vector $\bm{x} = (x_{1}, x_{2}, \ldots ,x_{d})^\intercal$. Each neuron is visible and every pair of neurons can be connected with an undirected weight. Suppose that we are supplied with $M$ activation patterns $\bm{x}^{(m)}$ that consist the training data, which may be e.g. the pixel data of images. These activation patterns are summarized by the Hebbian weight matrix.
\begin{defi}\label{Def:HebWeight}
The normalized Hebbian weight matrix $W$ is defined as $W_{ij}=\frac{1}{Md}\sum_{m=1}^M x^{(m)}_i x^{(m)}_j$ for $i \neq j$ and $W_{ii} = 0$.
\end{defi}
Hence, $W$ is a square $d$-dimensional, real, and symmetric matrix. It encodes the simple principle from Hebb that neurons that fire together also wire together~\cite{hebb1949organization}. The normalization of $W$ is such that the operator norm satisfies $\Vert W \Vert = \Ord{1}$. More generally, we do not have to restrict the training data $\bm{x}^{(m)}$ to be activation patterns of neural networks, but can instead allow them to be arbitrary vectors corresponding to relevant data.

We can then define Hebbian learning.
\begin{defi}
Hebbian learning is the process of:
\begin{enumerate}
\item{Constructing the Hebbian weight matrix $W$ from the set of training data $\bm{x}^{(m)}$.}
\item{Using the Hebbian weight matrix within the operation of a machine learning algorithm.}
\end{enumerate}
\end{defi}

\subsection{Quantum Hebbian learning}

Quantum Hebbian learning (QHL) can be thought of as the quantum analogy to Hebbian learning.
\begin{defi}\label{defHebbian}
We define quantum Hebbian learning to include any quantum algorithm that performs the following:
\begin{enumerate}
\item Provides a quantum analogue representation of the Hebbian weight matrix $W$.
\item Allows access to this quantum representation and makes the contained data operationally available in a quantum algorithm.
\end{enumerate}
In addition, we call a quantum Hebbian learning method { \rm dimension efficient} if 1) and 2) require $\Ord{{\rm poly}(\log (d)) }$ quantum bits and operations.
We call a quantum Hebbian learning method { \rm data efficient} if 1) and 2) require $\Ord{{\rm poly}(\log (M)) }$ quantum bits and operations.
\end{defi}
The link to bcQSE can now be elucidated by considering learning with generic quantum training data $\{|x^{(m)}\rangle\}_{m=1}^{M}$.
\begin{theorem} \rm{(QHL using bcQSE)}
With access to multiple copies of quantum training data $\{|x^{(m)}\rangle\}_{m=1}^{M}$, QHL can be achieved through bcQSE.
\end{theorem}
\begin{proof}
We first note by referring to Definition~\ref{Def:HebWeight} that $\rho - \mathbb{I}_{d}/d$ forms the quantum analogue of the Hebbian weight matrix, where $\rho = (1/M)\sum_{m=1}^{M}|x^{(m)}\rangle \langle x^{(m)}|$ and with $\mathbb{I}_{d}$ the identity operator. Then, using Theorem~\ref{BatchedQSE} we see that $\rho$ is operationally accessible through the transformation $e^{-it \ket{1}\bra{1} \otimes \rho}$ acting on a learning qubit and $N$ processing qubits. Furthermore, the explicit operational availability of $W$ is guaranteed from $\rho$ since
\begin{equation}
e^{-i \Delta t \ket 1 \bra 1 \otimes W } = e^{-i \Delta t \ket 1 \bra 1 \otimes \rho} e^{i \Delta t \ket 1 \bra 1 \otimes \frac{\mathbb{I}_{d}}{d}} + \Ord{\Delta t^2},
\end{equation}
i.e. so that in every batch of bcQSE one can apply a conditional phase $e^{i \Delta t \ket 1 \bra 1 \otimes \frac{\mathbb{I}_{d}}{d}}$, which is equivalent to applying a single qubit time-dependent phase and Pauli $Z$ gate for each batch. Note that this step can typically be skipped as the incurred phase error is small $\Ord{\frac{1}{d}}$.
\end{proof}

Thus, we have shown that, from our definitions, QHL and batched quantum state exponentiation are closely related. The given error analysis of bcQSE applies similarly to QHL, and results in a \textit{dimension-efficient} simulation. Data efficiency is not achieved with the methods shown in this work.

\begin{theorem}
If $t = \Ord{{\rm poly} (N  M) }$ and $1/\epsilon = \Ord{{\rm poly}( N M)}$, then QHL realized through bcQSE is dimension efficient.
\end{theorem}
\begin{proof}
From Table~\ref{Table:Scaling}, using bcQSE to perform QHL with Clifford+$T$ gates is polynomial in $N$ for the required number of qubits and gate count.
\end{proof}

Let us elaborate further on the nature of the quantum training data. On the one hand, we can assume the training data to be inherently quantum, originating from another quantum device. Alternatively, we can attempt to encode classical training data into quantum states. Suppose one is given the $d$-dimensional training data $\bm{x}^{(m)}$ where without loss of generality  $d = 2^{N}$.
The Hebbian weight matrix is written as in the quantum case
\begin{equation} \label{Eq:Encoding}
W = \frac{1}{M d}  \left[ \sum_{m=1}^{M}\bm{x}^{(m)}\left(\bm{x}^{(m)}\right)^\intercal \right] - \frac{\mathbb{I}_{d}}{d}. 
\end{equation}
Assume the following oracle.
\begin{oracle}\label{oracle1}
We are given $M$ oracles such that the operation $\ket{0\cdots 0} \to \ket{\bm x^{(m)}}$ on $N=\lceil \log d \rceil$ qubits can 
be performed, with $|\bm{x}^{(m)}\rangle = \frac{1}{\sqrt{d}} \sum_{i} x_{i}^{(m)}\ket{i}$ and $m=1,\dots, M$. Each oracle has a runtime of at most $T_{\rm data}$.
\end{oracle}
Since the oracle generates quantum states as required for batched quantum state exponentiation, we have the following corollary.
\begin{cor} \label{Cor:clInputQHL}
Given Oracle \ref{oracle1},  classical-input QHL can be performed.
If $t = \Ord{{\rm poly}(N  M)}$, $1/\epsilon = \Ord{{\rm poly} (N M)}$, and $T_{\rm data}=\Ord{{\rm poly} (NM)}$ then classical-input QHL is dimension efficient. 
\end{cor}
Efficiency of classical-input QHL requires that the qubit training states $\ket{x^{(m)}}$ can  be produced efficiently. This is a topical question that requires future development, but has also not been conclusively ruled out~\cite{aaronson2015read,giovannetti2008quantum,soklakov2006efficient}.

\subsection{Phase estimation}

Here we focus on achieving QHL using bcQSE and elaborate further on the operational availability of $\rho$ through phase estimation, allowing for calculation of eigenvalues (and corresponding eigenvectors) of $\rho$.
\begin{theorem} {\rm (Phase estimation with bcQSE)} \label{corrPhaseEstimation}
An eigenvalue of $\rho$ can be estimated to error $\epsilon$ through phase estimation using $\Ord{\frac{1}{\epsilon}}$ applications of $e^{-i\rho}$. In the setting of Theorem \ref{thmArbError}, realizing bcQSE using Clifford+$T$ gates results in an overall overhead for phase estimation of:
\begin{itemize}
\item{$\Ord{\frac{1}{\epsilon^{3}}}$ batches of quantum data;}
\item{$\tOrd{\frac{MN}{\epsilon^{3}} \log \left( \frac{MN}{\epsilon^3} \right)}$ qubits;}
\item{$\tOrd{\frac{MN}{\epsilon^{3}} \log \left( \frac{MN}{\epsilon^3} \right)}$ Clifford+$T$ gates;}
\end{itemize}
\end{theorem}
\begin{proof}
Consider the spectral decomposition $\rho = \sum_{i}\lambda_{i}\ket{\lambda_{i}}\bra{\lambda_{i}}$ and suppose $N$ processing qubits are prepared in the eigenstate $\ket{\lambda_{i}}$. Using Kitaev's version Ref.~\cite{kitaev1995quantum,kimmel2017hamiltonian}, phase estimation allows $\lambda_{i}$ to be estimated to precision $\epsilon$ using $\Ord{\frac{1}{\epsilon}}$ controlled unitary applications of $e^{-i \rho}$. If the processing qubits are prepared in an arbitrary state $\ket{\psi} = \sum_{i}c_{i}\ket{\lambda_{i}}$, $\sum_{i}\left|c_{i}\right|^{2}=1$, then the eigenvalue $\lambda_{i}$ will be measured with probability $\left|c_{i}\right|^{2}$ and the resultant state of the processing qubits is $\ket{\lambda_{i}}$~\cite{nielsen2002quantum}.

Now, if each application of bcQSE for time $t=1$ can be achieved to error $\epsilon^2$, then the overall simulation error after $\Ord{\frac{1}{\epsilon}}$ repetitions is $\Ord{\epsilon}$. The number of steps required in any single bcQSE interaction for phase estimation is thus $n=\Ord{\frac{1}{\epsilon^2}}$. Referring to Table~\ref{Table:Scaling}, we see that all the bcQSE operations combined require $\Ord{\frac{1}{\epsilon^{3}}}$ batches of quantum data, $\tOrd{\frac{MN}{\epsilon^{3}} \log \left( \frac{M N}{\epsilon^3} \right)}$ qubits and gates.

Phase estimation also has an overhead. It requires $\Ord{\frac{1}{\epsilon}}$ additional qubits (referred to here as the learning qubits) containing the eigenvalues and an additional $\Ord{\frac{1}{\epsilon^{2}}}$ gates. Both of these numbers are dominated by their counterparts from bcQSE.
\end{proof}
Phase estimation is the prototypical operational usage of $\rho$, leading us to the following definition.
\begin{defi}
Quantum Hebbian eigensystem learning is any quantum algorithm that uses QHL to derive the eigenvalues and/or eigenvectors of the corresponding Hebbian weight matrix.
\end{defi}
We note that the above Corollary \ref{Cor:clInputQHL} implies that classical-input quantum Hebbian eigensystem learning can be dimension efficient using the encoding in Eq.~(\ref{Eq:Encoding}).

\section{Discussion}\label{sec:con}

This work presents a method of quantum state exponentiation by using batches of pure quantum data (bcQSE). Here, the controlled partial swap between processing and data qubits acts as the elementary transformation which must be repeated multiple times for varying registers of data qubits. We have decomposed  this elementary transformation into Clifford+$T$ gates, hence allowing for the realization of bcQSE through a compilation of Clifford+$T$ gates. An analysis of the efficiency of our decomposition was also provided. It must be emphasized that this decomposition is not necessarily optimal, and it would be of interest to compare to the results of quantum compilers $-$ which can aim to minimize the gate cost overhead~\cite{heyfron2017efficient}.

We have presented bcQSE within the context of quantum Hebbian learning. A general formalism for quantum Hebbian learning was constructed based upon ideas established in classical Hebbian learning. However, the application of bcQSE is not restricted to classical-input quantum Hebbian learning: a neuroscience analogy allows bcQSE to be thought of representing multiple levels of perception (as in the human brain.) In this analogy, the data qubits represent the environment, which is continually interacting with processing qubits acting as low-level senses. The learning qubits are then the higher levels of perception, using information from the processing qubits to learn about the environment. We further extended the concept of Hebbian learning to quantum phase estimation, showing how the elementary controlled partial swap operation can be built up to achieve important quantum algorithms with applicability for quantum learning.

Quantum Hebbian learning holds relevance for the teaching of quantum and classical neural networks, such as the Hopfield network (and its quantum versions~\cite{rebentrost2017quantum}), with application in pattern recognition and optimization. Future work can consider quantum analogues of more advanced classical learning techniques such as the Storkey learning rule~\cite{storkey1999basins}, or more general applications of the batching process in machine learning.

\acknowledgements
We acknowledge Seth Lloyd, Iman Marvian and George Siopsis for insightful discussions.

\bibliographystyle{apsrev}
\bibliography{QML2}

\end{document}